\documentclass[runningheads]{fmcad}
\usepackage{graphicx}
\usepackage{listings}
\usepackage[dvipsnames]{xcolor}
\usepackage{amsmath}
\usepackage[english]{babel}
\usepackage[misc]{ifsym}
\usepackage{subfloat}
\usepackage{float}
\usepackage{amssymb}
\usepackage{amsthm}
\usepackage{cite}

\newif\ifbadgeavailable\newif\ifbadgefunctional\newif\ifbadgereusable
\badgeavailabletrue\badgefunctionaltrue\badgereusablefalse
\usepackage{blkarray}
\usepackage{tikz}

\usetikzlibrary{external}
\usepackage[normalem]{ulem}
\usepackage[ruled,linesnumbered,noend]{algorithm2e}
\usepackage[inline]{enumitem}
\definecolor{codegreen}{rgb}{0,0.6,0}
\definecolor{codegray}{rgb}{0.5,0.5,0.5}
\definecolor{codepurple}{rgb}{0.58,0,0.82}
\definecolor{backcolour}{rgb}{0.95,0.95,0.92}

\lstdefinestyle{mystyle}{
    backgroundcolor=\color{backcolour},
    commentstyle=\color{codegreen},
    keywordstyle=\color{magenta},
    numberstyle=\tiny\color{codegray},
    stringstyle=\color{codepurple},
    basicstyle=\ttfamily\footnotesize,
    breakatwhitespace=false,
    breaklines=true,
    captionpos=b,
    keepspaces=true,
    numbers=left,
    numbersep=5pt,
    showspaces=false,
    showstringspaces=false,
    showtabs=false,
    tabsize=2
}

\usetikzlibrary{shapes, tikzmark, arrows, calc, positioning, intersections, automata, fit}
\usetikzlibrary{decorations.pathreplacing, arrows.meta}
\usetikzlibrary{decorations.pathmorphing}

\newtheorem{lemma}{Lemma}
\theoremstyle{definition}
\newtheorem{definition}{Definition}

\newtheorem{example}{Example}

\begin{document}

% \SetWatermarkText{%
% \raisebox{-4cm}% AUTHORS: Move the badge up/down; should be next to title, authors, or affiliations
% {\ifbadgeavailable\includegraphics[width=11mm]{Figures/FM_2024_AE_available}\hspace{.815\linewidth}\else\hspace{.905\linewidth}\fi%
% \ifbadgefunctional\includegraphics[width=11mm]{Figures/FM_2024_AE_functional}\fi}}

\title{Termination analysis with interpolation-based transition invariant generation}

% \titlerunning{Safety-based Termination Analysis With Transition Invariant Generation}
%
% \titlerunning{Reachability Analysis for Multiloop Programs Using TPA}
% \author{Konstantin Britikov(\Letter)\inst{1}
%  \orcidID{0009-0005-7843-7290}
% \and
% Natasha Sharygina\inst{1}
%  \orcidID{0000-0002-8872-4913}
% \and
% Grigory Fedyukovich\inst{1,2}
%  \orcidID{0000-0003-1727-4043}
% \and
% Martin Blicha\inst{1}
%  \orcidID{0000-0001-8140-4098}
% }

\author{\IEEEauthorblockN{Konstantin Britikov}
\IEEEauthorblockA{USI, Switzerland\\
% Email: britik@usi.ch
}
\and
\IEEEauthorblockN{Martin Blicha}
\IEEEauthorblockA{Argot Collective, Switzerland\\
% Email: martin.blicha@argot.org
}
\and
\IEEEauthorblockN{Grigory Fedyukovich}
\IEEEauthorblockA{FSU, United States\\
% Email: grigory.fedyukovich@gmail.com
}
\and
\IEEEauthorblockN{Natasha Sharygina}
\IEEEauthorblockA{USI, Switzerland\\
% Email: sharygin@usi.ch
}}

% \author{
% \IEEEauthorblockN{Konstantin Britikov \and
%                   Martin Blicha \and
%                   Grigory Fedyukovich \and 
%                   Natasha Sharygina}
    
% %  \IEEEauthorblockA{
% %     \IEEEauthorrefmark{1}
% %   \textit{University of Lugano}\\
% %   Lugano, Switzerland
% % }
% %  \IEEEauthorblockA{
% %   \IEEEauthorrefmark{2}
% %   \textit{Argot Collective}\\
% %   Zug, Switzerland
% % }
% % \IEEEauthorblockA{
% %   \IEEEauthorrefmark{3}
% %   \textit{Florida State University}\\
% %   Tallahassee, United States
% % }

% }

% \institute{}
% \authorrunning{Britikov K., Blicha M., Fedyukovich G., Sharygina N.}
% \institute{University of Lugano, Lugano, Switzerland \and
% Florida State University, Tallahassee, FL, United States
 % \url{http://www.springer.com/gp/computer-science/lncs}

%
\maketitle              % typeset the header of the contribution

\begin{abstract}
Termination and nontermination of infinite-state systems are complementary problems that, 
despite their close connection, are typically addressed by separate techniques.
The core idea of this paper is to connect termination and nontermination analysis,
enabling the two to share intermediate results and guide one another.
We present a new termination analysis approach based on the
generation of well-founded transition invariants. 
It leverages Craig interpolation for transition invariant 
generation, capturing the structural reasons for termination.
The proposed technique extends safety-based nontermination analysis, enabling it to prove both 
termination and nontermination within a unified framework.
We implemented our approach in the \textsc{Golem} verification framework and evaluated it on benchmarks from the 
Termination Competition (TermComp). 
Empirical results demonstrate that combining termination and nontermination is beneficial and 
yields performance comparable to state-of-the-art tools.
\end{abstract}

%

% \newcommand{\oreach}[1]{\ensuremath{\mathit{TPA}^{{\leq}#1}}}
% \newcommand{\oreachfull}[3]{\ensuremath{\oreach{#1}(#2,#3)}}
% \newcommand{\atrleq}[1]{\ensuremath{\mathit{TPA}^{{\leq}#1}}}
% \newcommand{\atreq}[1]{\ensuremath{\mathit{TPA}^{{=}#1}}}
% \newcommand{\atrlt}[1]{\ensuremath{\mathit{TPA}^{{<}#1}}}

% \newcommand{\trleq}[1]{\ensuremath{\mathit{R}^{{\leq}#1}}}
% \newcommand{\treq}[1]{\ensuremath{\mathit{R}^{{=}#1}}}
% \newcommand{\trlt}[1]{\ensuremath{\mathit{R}^{{<}#1}}}

% PRELIMINARIES  VARIABLES
\newcommand{\source}{\mathit{src}}
\newcommand{\target}{\mathit{trg}}
\newcommand{\constraint}{\mathit{constr}}
\newcommand{\clause}{\mathit{c}}
\newcommand{\clauses}{C}
\newcommand{\pred}{\mathit{p}}
\newcommand{\preds}{\mathcal{P}}
\newcommand{\predsf}{\mathit{preds}}
\newcommand{\system}{\mathit{S}}
\newcommand{\encoding}{\pi}
\newcommand{\node}{n}
\newcommand{\init}{\mathit{Init}}
\newcommand{\Init}{\mathit{Init}}
\newcommand{\route}{{{{h}}}}
\newcommand{\routes}{\mathit{H}}
\newcommand{\iteration}{{{\gamma}}}
\newcommand{\proute}{{{\eta}}}
\newcommand{\proutes}{\Theta}
\newcommand{\query}{\mathit{Target}}
\newcommand{\incoming}{\mathit{\alpha}}
\newcommand{\outgoing}{\mathit{\omega}}
\newcommand{\vars}{\mathcal{V}}
\newcommand{\TS}{\mathit{TS}}
\newcommand{\NT}{\mathit{NT}}
\newcommand{\tr}{\mathit{Tr}}
\newcommand{\Tr}{\mathit{Tr}}
\newcommand{\inv}{\mathit{Inv}}
\newcommand{\Inv}{\mathit{Inv}}
\newcommand{\trInv}{\mathit{TrInv}}
\newcommand{\TrInv}{\mathit{TrInv}}
\newcommand{\mop}{\mathit{MOP}}
\newcommand{\tpa}{\textsc{TPA}\xspace}
\newcommand{\res}{\mathit{res}}
\newcommand{\Cs}{\mathit{Cs}}
\newcommand{\Prs}{\mathit{Prs}}
\newcommand{\flas}{\mathit{Flas}}
\newcommand{\id}{\mathit{Id}}
\newcommand{\bad}{\mathit{Bad}}
\newcommand{\Sink}{\mathit{Sink}}
\newcommand{\varsx}{\ensuremath{\vec{x}}}

\newcommand{\tuple}[1]{\langle #1 \rangle}

\newcommand{\unrollings}{\mathcal{U}}
\newcommand{\summaries}{\ensuremath{\Sigma}}
\newcommand{\summary}{\ensuremath{\sigma}}

\newcommand{\false}{\mathit{false}}
\newcommand{\true}{\mathit{true}}

\newcommand{\QE}{\textsc{QE}}

\newcommand{\SAFE}{\textsc{SAFE}}
\newcommand{\UNSAFE}{\textsc{UNSAFE}}
\newcommand{\NONTERM}{\textsc{NONTERM}\xspace}
\newcommand{\TERM}{\textsc{TERM}\xspace}
\newcommand{\UNKNOWN}{\textsc{UNKNOWN}\xspace}
\newcommand{\SAT}{\textsc{SAT}}
\newcommand{\UNSAT}{\textsc{UNSAT}}

\newcommand{\DAG}{D}
\newcommand{\curr}{\mathit{curr}}

\newcommand{\gri}[1]{$[$\textbf{\textcolor{BrickRed}{Grigory}}:~~\emph{\textcolor{Mahogany}{#1}}$]$}
\newcommand{\martin}[1]{$[$\textbf{\textcolor{BrickRed}{Martin}}:~~\emph{\textcolor{blue}{#1}}$]$}
\newcommand{\konst}[1]{$[$\textbf{\textcolor{ForestGreen}{Konstantin}}:~~\emph{\textcolor{blue}{#1}}$]$}

% ALGORITHM

\newcommand{\cexp}{\ensuremath{\mathit{cexp}}}

\newcommand{\analyze}{\textsc{Analyze}\xspace}
\newcommand{\blockTerm}{\textsc{BlockTerm}\xspace}
\newcommand{\trInvSynthesis}{\textsc{TrInvGen}\xspace}
\newcommand{\safetyCheck}{\textsc{SafetyCheck}\xspace}
\newcommand{\toDNF}{\textsc{ToDNF}\xspace}

\newcommand{\inc}[1]{\ensuremath{{#1}^{\mathit{pre}}}}
\newcommand{\out}[1]{\ensuremath{{#1}^{\mathit{post}}}}
\newcommand{\checked}{\mathit{curr}}
\newcommand{\safe}[1]{\ensuremath{#1.\mathit{sat}}}
\newcommand{\mbp}{\ensuremath{\mathit{MBP}}\xspace}
\newcommand{\splittpa}{\textsc{TPA}}
\newcommand{\reached}{\ensuremath{\mathit{visited}}}
\newcommand{\blocked}{\ensuremath{\mathit{blocked}}}
\newcommand{\current}{\ensuremath{\mathit{curr}}}
\newcommand{\outp}{\ensuremath{\mathit{output}}}
\newcommand{\cex}{\ensuremath{\mathit{cex}}}
\newcommand{\iter}{\ensuremath{\mathit{depth}}}
\newcommand{\iterConst}{\ensuremath{\mathit{iterConst}}}
\newcommand{\iterReached}{\ensuremath{\mathit{iterVisited}}}
\newcommand{\ex}{\ensuremath{\mathit{e}}}

\newcommand{\pop}{\ensuremath{\mathit{pop}}}
\newcommand{\learnedConst}{\mathit{constr}}
\newcommand{\precond}{\mathit{pre}}
\newcommand{\postcond}{\mathit{post}}
\newcommand{\sourcepred}{\mathit{s}}
\newcommand{\targetpred}{\mathit{t}}

\newcommand{\model}{M}
\newcommand{\algorithmicbreak}{\textbf{break}}

\newcommand{\sat}{\textsc{SAT}}
\newcommand{\unsat}{\textsc{UNSAT}}

\section{Introduction}
\label{Sec:Introduction}

Termination analysis of infinite-state systems is a fundamental problem in formal methods. 
% that necessitates the discovery of 
% complex ranking functions or inductive invariants to over-approximate system behavior.
% Konst: It not necessarily necessitates the discovery of complex ranking functions
Proving termination and nontermination is undecidable~\cite{DBLP:conf/cav/Braverman06}, 
yet there exist many techniques that succeed on a large class of real-world problems.
Modern techniques treat termination and nontermination as separate 
problems, since each has a different proof procedure.
%State-of-the-art termination techniques are based on the synthesis of 
%ranking functions for the transition relation
The existence of a ranking function by e.g.~\cite{DBLP:conf/vmcai/PodelskiR04,DBLP:conf/cav/KuraUH20,DBLP:journals/jacm/Ben-AmramG14,DBLP:conf/cav/HeizmannHP14,DBLP:journals/jar/GieslABEFFHOPSS17,DBLP:conf/esop/SaritaSGSV26,DBLP:conf/sas/Ben-AmramDG19,DBLP:conf/tacas/UrbanGK16,FedyukovichZG18,RileyF25} 
witnesses termination by providing a measure that strictly decreases along every transition.
Alternatively, fixpoint computation, e.g.~\cite{DBLP:journals/pacmpl/UnnoTGK23}, and disjunctively well-founded 
transition invariants~\cite{Transition-Invariants-Revisited-CAV26,DBLP:conf/pldi/CookPR06,DBLP:conf/cav/KroeningSTW10,DBLP:conf/tacas/TsitovichSWK11} have been proposed for termination analysis.
Nontermination techniques are predominantly based on the detection of recurrent 
sets — sets of states from which execution cannot escape~\cite{DBLP:conf/sas/Ben-AmramDG19,DBLP:conf/cade/FrohnG22,DBLP:conf/tacas/BrockschmidtCIK16,DBLP:conf/foveoos/BrockschmidtSOG11,DBLP:conf/cav/LarrazNORR14,FedyukovichZG18}.
There also exist alternative approaches based on Hoare-style reasoning~\cite{DBLP:conf/pldi/LeQC15}, 
safety proving~\cite{DBLP:conf/tacas/ChenCFNO14,DBLP:journals/entcs/BiereAS02}, representing infinite 
runs as sums of geometric series~\cite{DBLP:conf/tacas/LeikeH18},
or synthesis of \emph{negative} exact loop bounds~\cite{RileyF25}.

In this work, we present a unified framework for proving both termination and nontermination.
This framework extends the safety-based nontermination 
analysis~\cite{DBLP:conf/tacas/ChenCFNO14} (SNA), enabling it to prove termination.
The extension is a novel interpolation-based method for generating
disjunctively well-founded transition invariants for termination analysis.
The proposed approach uses terminating transition traces constructed by SNA, overapproximating them with 
a Craig interpolation~\cite{Craig_1957}.
The interpolant is processed into a transition invariant candidate, which is used to prove termination.
Additionally, based on the transition invariant candidate, the algorithm creates smaller, focused
(non)termination problems, which guide the analysis of the whole transition system.
Solving these problems helps to extend the transition invariant candidate and
enables an efficient nontermination analysis of complex problems.

We implemented our new technique in the \textsc{Golem} Constrained Horn Clause solver~\cite{DBLP:conf/cav/BlichaBS23}. 
We evaluated it on the Integer Transition Systems benchmarks from the Termination Competition~\cite{DBLP:conf/cade/GieslMRTW15}. 
Our tool solved a large number of terminating instances and demonstrated significant improvement for 
nontermination analysis, compared to the classical SNA.
The comparison with the winners of the Termination Competition~2025\footnote{\url{https://termination-portal.org/wiki/Termination\_Competition\_2025}.}---\textsc{KoAT}, \textsc{LoAT}~\cite{10.1007/978-3-031-90660-2_13}, and \textsc{T2}~\cite{DBLP:conf/tacas/BrockschmidtCIK16}---
shows that our approach is competitive with the state-of-the-art tools and can solve unique instances, including those
that were never solved by any tool before.

Overall, the contributions of this paper are as follows: (1) a novel interpolation-based method for the generation of 
disjunctively well-founded transition invariants, (2) a unified termination/nontermination analysis procedure that utilises 
the intermediate results of (non)termination analysis for self-guidance and (3) its 
implementation and successful evaluation.

% ZCBMEuOLDC0vfaFBU2k9mGuhR9_jp6my075v0pJmoqU

%

\section{Preliminaries}
\label{Sec:Preliminaries}
This section introduces the definitions and formalisms used in the paper. 
Logical formulas in our approach are defined over the theory of Linear Integer Arithmetic (LIA).
We use $X,X'$ to denote sets of variables (unprimed and primed), where $X$ is a set of state variables, and $X'= \{x'| x \in X\}$ is a set of next state variables.
If a formula contains variables over multiple states, we use superscript notation, e.g., $X^{(0)}, \dots, X^{(n)}$.

The formula defined over $X$ is a state formula, and a formula over $X \cup X'$ is a transition formula.
State formulas encode a set of states, and transition formulas define binary relations over the states.
We use $Id(X,X')$ to denote the identity transition formula $\bigwedge_{x \in X} x = x’$.

%\textit{Transition systems and Safety Problem:} 
A \emph{transition system} is defined as $\TS = \tuple{ \Init, \Tr, X }$, where $\Init$ is a formula 
encoding the initial states of the system, $\Tr$ is a transition formula, encoding the transition 
relation, and $X$ is the set of system variables.
With respect to transition system, a \emph{transition trace} is a formula encoding multiple consecutive 
transitions: $\Tr^{n}(X^{(0)}, \dots, X^{(n)}) = \Tr(X^{(0)}, X^{(1)}) \land \dots \land \Tr(X^{(n-1)}, X^{(n)})$.
%We write $\Tr^{n}(X^{(0)},\dots,X^{(n)})$ for this conjunction throughout the paper.
A \emph{safety problem} is $P = \tuple{ \Init, \Tr, Bad, X }$, where $Bad$ is a formula representing 
states that violate the safety property.

Safety verification amounts to determining whether there exists a transition trace that reaches a $Bad$ state from some initial state.
If $\TS$ is safe, then there exists an inductive invariant $Inv(X)$ such that (1) $\forall X. \Init(X) \rightarrow Inv(X)$, 
(2) $\forall X, X'. Inv(X) \land \Tr(X,X') \rightarrow Inv(X')$, and (3) $\forall X. Inv(X) \rightarrow \lnot Bad(X)$.
We use \safetyCheck$(P)$ to denote a safety verification procedure that returns a tuple $\tuple{\SAFE, Inv(X)}$ if 
$\TS$ is safe ($Inv(X)$ is an inductive invariant).
% Otherwise, if \TS is unsafe it returns $\tuple{\UNSAFE, \Tr^{n}(X^{(0)}, \dots, X^{(n)})}$, such that: 
If $\TS$ is unsafe it returns $\tuple{\UNSAFE, \Tr^{n}(X^{(0)}, \dots, X^{(n)})}$, s.t. $\exists X^{(0)}, \dots, X^{(n)}$\footnote{It is a shorthand notation for $\exists  X^{(0)}, \dots, \exists X^{(n)}$}.

$$\Init(X^{(0)}) \land \Tr^{n}(X^{(0)}, \dots, X^{(n)}) \land Bad(X^{(n)})$$

% \textit{Nontermination:}
A transition system is called \textbf{nonterminating} if there exists an initial state from which it is possible to take 
any number of transitions, meaning $\exists X. \Init(X)$ such that:
$$\forall n \in \mathbb{N}.  \exists X^{(1)}, \dots, X^{(n)}. \Tr^{n}(X, X^{(1)}, \dots, X^{(n)})$$

\begin{definition}[Recurrent Set]
    For a Transition System $\tuple{ \Init, \Tr, X }$, a formula $N(X)$ is a recurrent set if: 
    (1) $\forall X. N(X) \rightarrow \exists X'. \Tr(X,X') \land N(X')$,  
    (2) $\exists X. N(X) \land \Init(X)$.
\end{definition}

Transition System is nonterminating if there exists a \emph{recurrent set}. 
Intuitively, every state in a recurrent set has a successor that also lies in the recurrent set, and the recurrent set contains some of the 
initial states.

% \textit{Termination:} 
A $\TS$ is called \textbf{terminating} if for every initial state it enables only a finite number of transitions.
% For a Transition System $\tuple{ \Init, \Tr, X }$:
Transition System $\tuple{ \Init, \Tr, X }$ terminates, if $\forall X. \Init(X)$ the following holds:
% $$\forall X. \Init(X) \rightarrow \exists n \in \mathbb{N}. \lnot \exists \dots, X^{(n)}. \Tr^{n}(X, X^{(1)}, \dots, X^{(n)})$$
$$\exists n \in \mathbb{N}. \lnot \exists X^{(1)}, \dots, X^{(n)}. \Tr^{n}(X, X^{(1)}, \dots, X^{(n)})$$
% Termination can be proven by synthesizing a \emph{ranking function}~\cite{DBLP:conf/vmcai/PodelskiR04}, mapping states to a well-founded domain, such that its value strictly decreases along every transition.

\begin{definition}[Transition invariant]
    For a Transition System $\tuple{ \Init, \Tr, X }$, $\TrInv(X,X')$ is a transition invariant if: $\forall X,X'. R(X) \land \Tr^+(X,X') \rightarrow \TrInv(X,X')$, where $\Tr^+(X,X')$ is the transitive closure of $\Tr(X,X')$
    and $R(X)$ is a formula encoding all states reachable in $\TS$.
\end{definition}

\begin{definition}[Well-founded Relation]
    A binary relation $W \subseteq D \times D$ over a domain $D$ is \emph{well-founded} if there exists no infinite descending chain:
    $$\nexists\, d_0, d_1, d_2, \dots \in D.\ 
      \forall i \in \mathbb{N}.\ (d_i, d_{i+1}) \in W$$
\end{definition}

A relation $W(X,X')$ can be proven to be well-founded by synthesizing a ranking function $f(X)$, such that
 $\forall X,X'. W(X,X') \rightarrow f(X) > f(X')$~\cite{DBLP:conf/vmcai/PodelskiR04}.

Termination can be proven using transition invariants~\cite{DBLP:conf/lics/PodelskiR04}.
If a transition invariant is a finite union of well-founded relations (for example, if each relation admits a ranking function), 
then the transition system terminates.
\begin{definition}[Disjunctively well-founded relation]
    A relation $K(X,X') = \bigcup_{i = 0}^{n} W_i(X,X')$ is called disjunctively well-founded if every relation $W_i(X,X')$ is well-founded. 
\end{definition}

% \textit{Craig Interpolation:} 
Given two formulas $A$ and $B$ where $A \land B$ is unsatisfiable, a Craig interpolant~\cite{Craig_1957} 
for $A$ and $B$ is a formula $I$, where: (1) $A \rightarrow I$; 
(2) $I \land B$ is unsatisfiable; and (3) $I$ only contains the symbols, common to both $A$ and $B$.
Craig interpolants are widely used as overapproximations.
We write $\textsc{Itp}(A, B)$ for the computation of a Craig interpolant.

\subsection{Proving Nontermination via Safety}
Many techniques reduce termination and nontermination analysis to safety verification.
One notable approach is the reduction of nontermination to safety, which constructs
recurrent sets via inductive invariants~\cite{DBLP:conf/tacas/ChenCFNO14}.
It was introduced in the context of proving nontermination of programs with loops, and can be extended to transition systems.

\begin{definition}[Sink states]
    For a Transition System $\tuple{\Init, \Tr, X }$, a state $s$ is a sink state, if $\lnot \exists s'. \Tr(s,s')$.
    We use $\Sink(X)$ to denote the formula encoding all sink states in the transition system:
    $\forall X. \Sink(X) \leftrightarrow \lnot \exists X'. \Tr(X,X')$
\end{definition}

\begin{figure}[t]
    \centering
    \includegraphics[width=\linewidth]{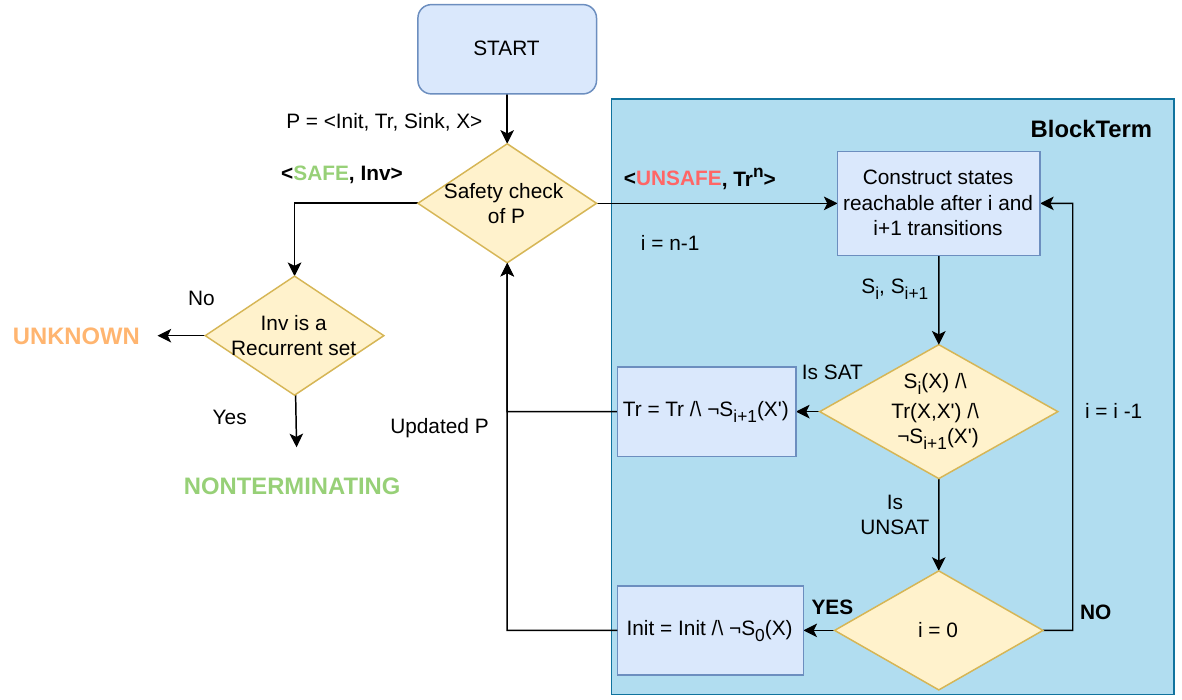}
    \caption{Execution flow of Safety-based Nontermination analysis, cf.~\cite{DBLP:conf/tacas/ChenCFNO14}}
    \label{fig:SNA}
\end{figure}

Algorithm~\ref{Alg:SNA} for Safety-based Nontermination analysis (SNA) takes a 
safety problem $P = \tuple{\Init, \Tr, \Sink, X}$, where sink states are constructed from $\TS$ using
quantifier elimination (\QE): $\Sink(X) = \lnot \QE(\exists X'. \Tr(X,X'))$.
As an output, it returns \NONTERM{} if the system is nonterminating, and \UNKNOWN, if the algorithm cannot conclude the analysis. 
Figure~\ref{fig:SNA} illustrates the execution flow of Safety-based Nontermination.

\begin{algorithm}[t!]
	\SetInd{0.5em}{0.5em}
    \DontPrintSemicolon
    \small
    \newcommand{\Continue}{\textbf{continue}}
    \caption{Safety-based Nontermination Analysis (SNA), cf.~\cite{DBLP:conf/tacas/ChenCFNO14}}
    \label{Alg:SNA}
    \SetKwInOut{Input}{Input}
    \SetKwInOut{Output}{Output}
    \Input{Safety problem $P = \tuple{\Init, \Tr, \Sink, X}$}
    \Output{\NONTERM $\mid$ \UNKNOWN}
    % \SetKwFunction{FAnalyze}{Analyze}
    % \SetKwProg{Fn}{Function}{:}{}
    % \nonl\Fn{\FAnalyze{$P$}}{
        \While{$\top$}{
            $\tuple{\res,\phi = \Tr^{n}\ or\, \Inv} \gets \safetyCheck(P)$\\
            \eIf(\label{SNA:UNSAFE}){$\res = \UNSAFE$} {%
                $P \gets \blockTerm(P, \phi)$ \label{SNA:BLOCKTERM}
            }
            {
                \If{$\forall X. \phi(X) \rightarrow \exists X'. \Tr(X,X')$\label{SNA:SAFE}}{
                    \Return \NONTERM \label{SNA:RECURRENT}
                }
                \Return \UNKNOWN
            }
            
        }
    % }
    
\end{algorithm}
% A

\iffalse
The approach executes a safety check of $P$ in a loop.
At line~\ref{SNA:UNSAFE}, the algorithm handles the case when $P$ is unsafe (i.e., there exists a trace reaching sink states).
In this case, the algorithm calls the $\blockTerm$ procedure (Algorithm~\ref{Alg:SNABlock}), which 
refines the safety problem $P$ by blocking the states that deterministically
lead to the sink states, updating the safety problem.
\fi
SNA iteratively solves increasingly refined safety problems until the given sink states are proven unreachable.
In every iteration, if there exists a trace reaching sink states (line~\ref{SNA:UNSAFE}), $\blockTerm$ (Algorithm~\ref{Alg:SNABlock}) refines the safety problem $P$ by blocking the states that deterministically
lead to the sink states, updating the safety problem.
$\blockTerm$ is described in detail below.
%When $P$ is safe, the algorithm checks if the inductive invariant $Inv(X)$ is a recurrent set (line~\ref{SNA:SAFE}).
When the original sink states are no longer reachable, SNA checks if the computed inductive invariant $\inv(X)$ is a recurrent set (line~\ref{SNA:SAFE}).
%The check at line~\ref{SNA:SAFE} is executed using quantifier elimination, checking if the following formula is unsatisfiable:
This amounts to deciding satisfiability of the following formula:
$$\inv(X) \land \lnot \exists X'. \Tr(X,X').$$
If $\inv(X)$ is a recurrent set, then the transition system is nonterminating.
Otherwise, SNA returns \UNKNOWN.

\begin{algorithm}[t!]
    \SetInd{0.5em}{0.5em}
    \DontPrintSemicolon
    \small
    \newcommand{\Continue}{\textbf{continue}}
    \caption{\blockTerm(P, $\Tr^{n}$), cf.~\cite{DBLP:conf/tacas/ChenCFNO14}}
    \label{Alg:SNABlock}
    \SetKwInOut{Input}{Input}
    \SetKwInOut{Output}{Output}
    \SetKwInOut{Data}{Data}
    \Input{Safety problem $P = \tuple{\Init, \Tr, \Sink, X}$, terminating trace $\Tr^{n}(X^{(0)}, \dots, X^{(n)}) = \Tr(X^{(0)}, X^{(1)}) \land \dots \land \Tr(X^{(n-1)}, X^{(n)})$}
    \Output{Refined safety problem $P'$}
    \Data{$S_0,\dots,S_n$ - formulas which encode states within trace reachable in 0, ..., n transitions}

    % \SetKwFunction{FBlock}{BlockTerm}
    % \SetKwProg{Fn}{Function}{:}{}
    % \Fn{\FBlock{$P, trace$}}{
    \SetAlgoFuncName{BlockTerm}{autoref name}
        $i \gets n-1$ \\
        \While{$i > 0$}{
            $S_{i+1}(X) \gets \QE(\exists X^{(0)},\dots, X^{(i)},  X^{(i+2)}, \dots, X^{(n)}.$ $\Init(X^{(0)}) \land \Tr^{n} \land \Sink(X^{(n)}))[X^{(i+1)} \mapsto X]$ \label{SNA:Si1}\\
            $S_{i}(X) \gets \QE(\exists X^{(0)},\dots, X^{(i-1)},  X^{(i+1)}, \dots, X^{(n)}.$ $\Init(X^{(0)}) \land \Tr^{n} \land \Sink(X^{(n)}))[X^{(i)} \mapsto X]$ \label{SNA:Si}\\
            \If{SAT ? $S_i(X) \land \Tr(X,X') \land \lnot S_{i+1}(X')$ \label{SNA:BLOCK}}{
                $\Tr'(X,X') \gets \Tr(X,X') \land \lnot S_{i+1}(X')$ \label{SNA:TRBLOCK}  \\
                \Return $\tuple{\Init, \Tr', \Sink, X}$
            }
            $i \gets i - 1$   
        }
        $\Init'(X^{(0)}) \gets \Init(X^{(0)}) \land \lnot \QE(\exists X^{(1)}, \dots, X^{(n)}. \Init(X^{(0)}) \land \Tr^{n} \land \Sink(X^{(n)}))$  \label{SNA:BLOCKINIT} \\
        \Return $\tuple{\Init', \Tr, \Sink, X}$
    % }
\end{algorithm}

The $\blockTerm$ procedure takes as input a safety problem $P$ and a trace formula $\Tr^{n}(X^{(0)}, \dots, X^{(n)})$ such 
that $\Init(X^{(0)}) \land \Tr^{n}(X^{(0)}, \dots, X^{(n)}) \land \Sink(X^{(n)})$ is satisfiable.
The algorithm iteratively analyses the trace from the last transition backwards. 
% It checks if it is possible to take a transition that does not lead to the sink states, starting from the end of the trace.
Using \QE, Algorithm~\ref{Alg:SNABlock} computes $S_i$ and $S_{i+1}$, 
encoding the states that are reachable after $i$ and $i+1$ transitions from $\Init$ (lines~\ref{SNA:Si1},~\ref{SNA:Si}).
Then, at line~\ref{SNA:BLOCK} it checks if $S_i(X) \land \Tr(X,X') \land \lnot S_{i+1}(X')$ is satisfiable.%, sending a query to the SMT solver.
%With this query, the algorithm checks if it is possible to reach some state that does not necessarily lead to the sink states.
Intuitively, with this query, the algorithm checks if it is possible to take an alternative transition at step $i$, one that does not lead to the sink states.
If it is possible, the algorithm blocks the states deterministically leading to the sink states by excluding them from the transition relation (line~\ref{SNA:TRBLOCK}).
If the whole trace is deterministic (it is not possible to take an alternative transition along the trace), \blockTerm blocks the subset of initial states deterministically leading to sink (line~\ref{SNA:BLOCKINIT}).

\begin{lemma}\label{Lem:NONTERM}
    If Algorithm~\ref{Alg:SNA} returns \NONTERM, then the transition system is nonterminating, cf.~\cite{DBLP:conf/tacas/ChenCFNO14}.
\end{lemma}
\begin{example}\label{Ex:SNA}
    Consider a transition system $\TS = \tuple{\Init, \Tr, X}$ with 
    $X = \{x\}$, $\Init(X)   \triangleq \top$, and
    $\Tr(X, X') \triangleq (x' = x - 1 \lor x' = x - 2) \land x \neq 0$.
    The sink states are $\Sink(X) \triangleq x = 0$, and the initial 
    safety problem is $P = \tuple{\Init, \Tr, \Sink, x}$.
    This $\TS$ is nonterminating, since $x$ infinitely decreases.

    The execution of Algorithm~\ref{Alg:SNA} updates the initial states
    and the transition relation in two iterations of the algorithm, via \blockTerm: 
    $\Init(X) \;\gets\; x \neq 0$, and  $\Tr(X, X') \;\gets\; \Tr(X, X') \land x' \neq 0$.
    The safety check of the updated problem concludes \SAFE,
    with $Inv(X) \triangleq x \neq 0$. 
    $Inv(X)$ is a recurrent set, and algorithm returns \NONTERM.

    \textit{Motivation for transition invariant-based termination analysis.} Consider a 
    small modification (replacing $x\neq0$ with $x>0$): $\Tr(X, X') \triangleq 
    (x' = x - 1 \lor x' = x - 2) \land x > 0$. 
    This renders $\TS$ terminating since $x$ is strictly positive and decreases every step. 
    However, Algorithm~\ref{Alg:SNA} cannot detect termination by design, and 
    it would return \UNKNOWN. 

    Let us examine the execution of Algorithm~\ref{Alg:SNA} on this example.
    First, using \blockTerm, it blocks initial states that instantly terminate:
    $\Init(X) \;\gets\; x > 0$.
    In the second iteration, \safetyCheck$(P)$ returns \UNSAFE{} with a 
    trace $\Tr^1(X^{(0)}, X^{(1)})$ satisfying $\exists X^{(0)}, X^{(1)}. 
    \Init(X^{(0)}) \land \Tr^1(X^{(0)}, X^{(1)}) \land \Sink(X^{(1)})$.
    Our key insight is to generalise such terminating traces into a formula describing an 
    entire class of terminating behaviours. 
    If this generalisation yields a disjunctively well-founded 
    transition invariant over all reachable states, it witnesses 
    termination.
    Particularly, $\TrInv(X, X') \triangleq x' < x \land x > 0$ is a well-founded transition 
    invariant that proves the termination of $\TS$.
    Section~\ref{Sec:Main} presents the algorithm that automates this generalisation via 
    Craig interpolation.
\end{example}

\section{Termination Analysis via Interpolation-based Transition Invariant Generation}\label{Sec:Main}

We present a new termination analysis approach that leverages safety verification to generate a 
well-founded transition invariant via interpolation.
At each iteration, our procedure constructs a trace from the initial states to a sink state.
A candidate well-founded transition invariant is then constructed based on the overapproximation of the terminating trace.
If the candidate is valid for all reachable states (meaning it is a transition invariant), termination is concluded. 
The interpolation-based approach enables efficient generation of a transition invariant guided by the termination conditions.

\subsection{Interpolation-based Generation of Well-Founded Transition Invariants}

The core of our approach is \trInvSynthesis, a procedure for the automated 
interpolation-based generation of well-founded transition invariants.
Given a transition system and a terminating trace, % procedure 
it constructs a disjunctively well-founded approximation of the trace. 
If the resulting formula constitutes a transition invariant over all reachable states, 
it witnesses termination.

The procedure consists of five steps:
% (1) computing $T^n$, the formula encoding states that terminate within $n$ transitions.
% $T^n$ is used to construct the interpolant. It guarantees that no states other then $\Sink$ are reachable in $n$ transitions; 
(1) Computing $\NT^n$, the formula encoding states from which non-sink states 
can be reached within $n$ transitions.
$\NT^n$ is used to construct the interpolant. 
The negation of $\NT^n$ encodes the states from which $\Sink$ is guaranteed to be reached in at most $n$ transitions.
(2) Constructing an over-approximation of the terminating trace via interpolation. 
The usage of sink states allows the interpolant to capture the reason for termination.
(3) Converting the interpolant into an underapproximating DNF formula.
This enables efficient ranking function synthesis for each disjunct.
(4) Filtering disjuncts, preserving only disjuncts that encode well-founded relations. These disjuncts are used to construct a candidate transition invariant.
Only well-founded relations are preserved, as transition invariant needs to be disjunctively well-founded to witness termination. 
(5) Computing the formula $C(X)$, encoding states for which candidate transition invariant is valid.
Transition invariant needs to be valid only over the reachable states to conclude termination.
Therefore, if all reachable states satisfy $C(X)$, then $\TS$ is terminating.
The overall flow of the procedure is illustrated in Figure~\ref{fig:TrInvGen} and described in Algorithm~\ref{Alg:SNA_TrInv}.

\begin{figure}[t]
    \centering
    \includegraphics[width=\linewidth]{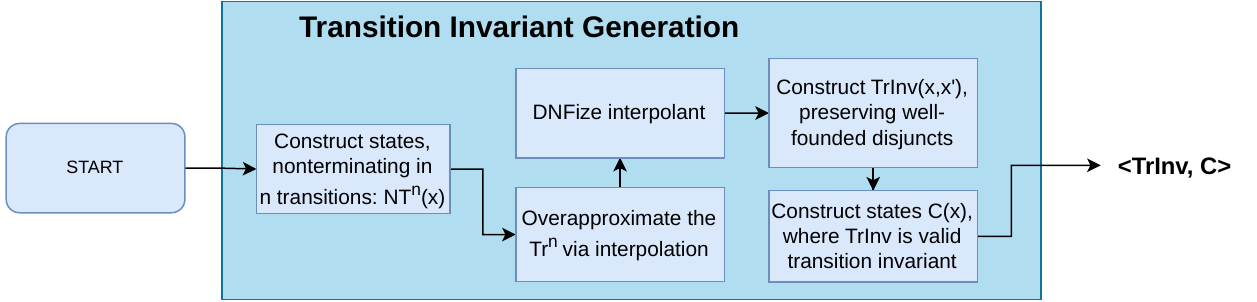}
    \caption{Execution Flow of Interpolation-based Transition Invariant Generation}
    \label{fig:TrInvGen}
\end{figure}

\begin{algorithm}[t!]
	\SetInd{0.5em}{0.5em}
    \DontPrintSemicolon
    \small
    \newcommand{\Continue}{\textbf{continue}}
    \caption{$\trInvSynthesis(P, \Tr^n, \TrInv)$.}
    \label{Alg:SNA_TrInv}
    \SetKwInOut{Input}{Input}
    \SetKwInOut{Output}{Output}
    \Input{Safety problem $P = \tuple{\Init, \Tr, \Sink, X}$, transition trace $\Tr^n(X^{(0)}, X^{(1)}, \dots, X^{(n)})$, candidate transition invariant $\TrInv(X,X')$}
    \Output{Updated transition invariant $\TrInv(X,X')$ and formula $C(X)$ encoding states covered by
    transition invariant}
    $\NT^n(X^{(0)}) \gets \QE(\exists X^{(1)}, \dots, X^{(n)}.$ $\Tr^{n}(X^{(0)}, X^{(1)}, \dots, X^{(n)}) \land \lnot \Sink(X^{(n)}))$ \label{SNA_TrInv:T} \\
    % $T^n(X^{(0)}) \gets \lnot \QE(\exists X^{(1)}, \dots, X^{(n)}.$ $\Tr^{n}(X^{(0)}, X^{(1)}, \dots, X^{(n)}) \land \lnot \Sink(X^{(n)}))$ \label{SNA_TrInv:T} \\
    % $itp(X,X') \gets Itp(\Tr^{n}(X, \dots, X'), T^n(X) \land \lnot \Sink(X'))$\label{SNA_TrInv:Itp}\\
    $itp(X,X') \gets \textsc{Itp}(\Tr^{n}(X, \dots, X'), \lnot \NT^n(X) \land \lnot \Sink(X'))$\label{SNA_TrInv:Itp}\\
    $\Cs(X,X') \gets \toDNF(itp)$ \label{SNA_TrInv:toDNF}\\
    $\TrInv \gets \TrInv \lor \bigvee_i \{d_i \in \Cs \mid \text{$\exists$ ranking function $f(X)$} \}$\label{SNA_TrInv:TrInv} \\
    $C(X) \gets Sink(X) \lor \lnot \QE(\exists X',X''.$ $\Tr(X,X') \land \bigl(Id(X',X'') \lor \TrInv(X',X'')\bigr) 
        \land \lnot \TrInv(X,X''))$ \label{SNA_TrInv:NC} \\
    \Return $\tuple{\TrInv, C}$
\end{algorithm}

Algorithm~\ref{Alg:SNA_TrInv} takes as input a safety problem $P = \tuple{\Init, \Tr, \newline\Sink, X}$,
a transition trace $\Tr^n(X^{(0)}, X^{(1)}, \dots, X^{(n)})$, and a transition invariant formula $\TrInv(X,X')$. 
It returns a tuple of formulas $\tuple{\TrInv, C}$, where $\TrInv(X,X')$ is a disjunctively well-founded transition invariant 
candidate, and $C(X)$ is a formula encoding the states for which $\TrInv$ is a valid transition invariant. 
It proceeds in the following steps:
% Define from the negation, not the itself

\textbf{Step 1}: \textit{Computing $\NT^n$ (line~\ref{SNA_TrInv:T}).} 
$\NT^n(X)$ encodes all states from which a non-sink state is reachable within $n$ transitions. 
The following property defines it: $\forall X^{(0)}. \NT^n(X^{(0)}) \leftrightarrow $
% $\NT^n(X)$ encodes all states guaranteed to reach a sink state within $n$ 
% transitions, is defined by the following property:
$$\exists X^{(1)}, \dots, X^{(n)}. \Tr^{n}(X^{(0)}, \dots, X^{(n)}) \land \lnot \Sink(X^{(n)})$$

$\NT^n$ is constructed with quantifier elimination, eliminating all variables from the formula except for $X^{(0)}$.
Construction of $\NT^n$ enables the overapproximation of the trace. 

\begin{example}\label{Ex:Tn}
    % Consider $\TS = \tuple{\Init, \Tr, X}$ with $\Init \triangleq  x > 0$, $\Tr(X,X') \triangleq  ((x' = x - 1 \lor x' = x - 2) \land x > 0) \lor (x < -10 \land x' = x - 1)$, $X \triangleq  \{x\}$.
    % For $n = 2$ quantifier elimination would yield $T^2(X) = x \leq 2$.
    % States with $x \in \{3,4\}$ can reach $\Sink$ in two transitions but are 
    % not guaranteed to do so, so they are not included in $T^2$.
    % States with $x \in \{1,2\}$ terminate in 1 or 2 transitions.
    Consider $\TS = \tuple{\Init, \Tr, X}$ with $\Init \triangleq  x > 0$, 
    $\Tr(X,X') \triangleq  (x' = x - 1 \lor x' = x - 2) \land (x > 0 \lor x < -10)$, $X \triangleq  \{x\}$.
    The sink states are $\Sink(X) \triangleq x \leq 0 \land x \geq -10$.
    For $n = 1$, quantifier elimination yields $\NT^1(X) = x > 1 \land x < -10$: 
    states from which $\lnot \Sink$ is reachable in one transition.
    % States with $x \in \{3,4\}$ can reach $\Sink$ in two transitions, but not necessarily,
    % therefore these states are included in $\NT^2$.
\end{example}

\textbf{Step 2}: \textit{Interpolation (line~\ref{SNA_TrInv:Itp}).}
The following formula is unsatisfiable by construction since any state that satisfies 
$\lnot \NT^n$ must reach $\Sink$ within at most $n$ steps:
$$\lnot \NT^{n}(X^{(0)}) \land \Tr^{n}(X^{(0)}, \dots, X^{(n)}) \land \lnot \Sink(X^{(n)})$$

This allows us to compute a \emph{Craig interpolant} that overapproximates the trace $\Tr^{n}$.
Crucially, the usage of the sink states for the interpolant construction allows the interpolant to capture the reason \emph{why} 
the particular trace is terminating.

\textbf{Step 3}: \textit{DNF underapproximation (line~\ref{SNA_TrInv:toDNF}).}
For a transition invariant to witness termination, it must be 
disjunctively well-founded: every disjunct must encode a 
well-founded relation. 
Converting the interpolant to DNF would simplify the well-foundedness
check for separate disjuncts.
However, complete DNF conversion can cause an exponential 
blowup in formula size and take a significant amount of time. 
Moreover, even a complete DNF is not guaranteed to yield a 
transition invariant.

Instead, our algorithm follows a lazy approach: it constructs an 
\emph{under-approximating} DNF formula on demand. 
Specifically, the algorithm constructs disjuncts as follows:
(1) the SMT solver produces a model that satisfies $itp$; (2) 
the literals satisfied by this model are extracted and 
conjoined to form a single disjunct;
(3) this disjunct is blocked from $itp$ and the process repeats 
until some $m$ disjuncts are constructed.
The result is a formula:
$$\Cs(X,X') \triangleq \bigvee_{i=0}^m d_i$$
where each $d_i$ is a conjunction, and $\Cs(X,X') \rightarrow itp(X,X')$.
This approach significantly reduces formula size and speeds up both the DNF
construction and subsequent well-foundedness check, while preserving correctness: 
in the worst case, $\Cs$ fails to be a transition invariant, but this is also 
possible even with a complete DNF conversion.
% Empirically, bounding the number of disjuncts produces more 
% compact and efficiently verifiable termination proofs.
% Our algorithm produces a model satisfying the SMT formula.
% Using produced model, it extracts the satisfied literals from the SMT formula, and constructs a conjunction of these literals.
% This conjunction is one of the dinsjuncts in the constructed DNF. This conjunction is blocked in the original formula, and the process proceeds.
% It allows to significantly reduce the size of the resulting formula and speed up the DNF conversion, while preserving correctness: 
% in the worst case, the resulting formula is not a transition invariant (which can happen even if the complete DNF is produced).
% Empirically, we observed that limiting the number of disjuncts results in more size-efficient and quick termination proofs.
% The execution of the procedure returns a formula:

% The number of disjuncts is bounded by a constant (a tunable parameter depending on the available size of RAM), since DNF 
% conversion can cause an exponential blowup in formula size.  Move to implementation & evaluation
% Bring citation 3 to background section

\textbf{Step 4}: \textit{Filtering Well-Founded Disjuncts (line~\ref{SNA_TrInv:TrInv}).}
For each disjunct $d(X,X') \in \Cs(X,X')$, the algorithm checks if $d(X,X')$ encodes a well-founded relation.
Our technique synthesizes linear ranking functions automatically following~\cite{DBLP:conf/vmcai/PodelskiR04} for checking the existence of a ranking function.
The disjuncts admitting a ranking function are preserved, the rest are discarded.
The resulting formula is disjoined with the input candidate transition invariant $\TrInv$.

\textbf{Step 5}: \textit{Synthesis of the states for which transition invariant is valid (line~\ref{SNA_TrInv:NC}).}
If $\TrInv$ is a transition invariant over all reachable states, then $\TS$ terminates, since every disjunct encodes a well-founded relation.
To verify this, the procedure computes the formula $C(X)$ encoding the states over which $\TrInv$ is a valid transition invariant (including sink states, as any candidate invariant is valid over $\Sink$).
% Split this thing into 2 properties
% Make it consistent with the code
\begin{definition}[Coverage formula]
    Given a transition system $\TS = \tuple{\Init, \Tr, X}$ and a formula $\TrInv(X,X')$, 
    a formula $C(X)$ is called coverage, if the following holds $\forall X, X',X''$:
    \begin{align*}
     C(X) \land \Tr(X,X') \rightarrow& \TrInv(X,X') \\
     C(X) \land \Tr(X,X') \land \TrInv(X',X'') \rightarrow& \TrInv(X,X'')
    \end{align*}
 That is, $\TrInv$ is a transition invariant over states that satisfy $C(X)$.
\end{definition}

Similarly to $\NT^{n}$, $C(X)$ is constructed by applying quantifier elimination, negating the states for 
which $\TrInv$ is not a transition invariant. It is a coverage formula by construction.
% example has an issue!
\begin{example}\label{Ex:C}
    Continuing Example~\ref{Ex:Tn}, suppose the procedure constructs an interpolant
    $itp(X,X') \triangleq x' < x \land (x > 0 \lor x < -10)$. 
    Afterwards, algorithm would construct DNF form: $(x' < x \land x > 0 ) \lor (x' < x \land x < -10 )$, and 
    since the second disjunct is not well founded, it would be dropped, resulting in $\TrInv(X,X') \triangleq x' < x \land x > 0$. 
    Quantifier elimination at line~\ref{SNA_TrInv:NC} then returns $C(X) \triangleq (x \geq -10 \land x \leq 0) \lor x > 0$.
\end{example}

\emph{Note.} If $\lnot C(X)$ is unreachable, then $\TrInv(X,X')$ is a disjunctively well-founded
transition invariant over all reachable states, and witnesses termination of the transition system.
For illustration, the transition invariant in Example~\ref{Ex:C} is sufficient to prove termination
since the states satisfying $\lnot C(X) \triangleq x < -10$ cannot be reached within $\TS$.

\begin{lemma}\label{Lem:Term}
    Let $\TS = \tuple{\Init, \Tr, X}$ be a transition system and let 
    $C(X)$ be a formula such that $\lnot C(X)$ is unreachable in 
    $\TS$. If $\TrInv(X,X')$ is a finite disjunction of formulas encoding well-founded
    relations and satisfies $\forall X, X',X''$:
    \begin{align*}
     C(X) \land \Tr(X,X') \rightarrow& \TrInv(X,X') \\
     C(X) \land \Tr(X,X') \land \TrInv(X',X'') \rightarrow& \TrInv(X,X'')
    \end{align*}    
    then $\TS$ is terminating.
\end{lemma}
\begin{proof}
    We show that $\TrInv$ is a disjunctively well-founded transition 
    invariant over the reachable states of $\TS$, from which 
    termination follows by~\cite{DBLP:conf/lics/PodelskiR04}.

    \textbf{Disjunctive well-foundedness.} By construction, for each disjunct 
    $d_i$ of $\TrInv = \bigvee_i d_i$ there exists a ranking function.
    Therefore, every disjunct is a logical encoding of 
    well-founded relation~\cite{DBLP:conf/vmcai/PodelskiR04}. 
    Hence $\TrInv$ is disjunctively well-founded.

    \textbf{Transition invariant.} We show that 
    $\forall X,X'.\ R(X) \land \Tr^+(X,X') \rightarrow \TrInv(X,X')$,
    where $R(X)$ denotes the reachable states of $\TS$, by induction 
    on the number of steps $n$ in $\Tr^+(X,X')$.

    \emph{Base case} ($n = 1$): It is needed to prove that:
    $$R(X) \land \Tr(X,X') \rightarrow \TrInv(X,X')$$ 
    Since $\lnot C(X)$ is unreachable, $R(X) \rightarrow C(X)$,
    and therefore $R(X) \land \Tr(X,X') \rightarrow C(X) \land \Tr(X,X')$. 

    The following assumption holds by the lemma statement:
    $$C(X) \land \Tr(X,X') \rightarrow \TrInv(X,X').$$
    Therefore the base case holds.

    \emph{Inductive step}: Assume $R(X) \land \Tr^n(X, \dots,X') \rightarrow 
    \TrInv(X,X')$ for some $n \geq 1$. 
    Let $\Tr^{n+1}(X,X'')$ hold, so there $\exists X, X', \dots, X''. R(X) \land \Tr(X, X') \land \Tr^n(X',\dots,X'')$. 
    Since $R(X)$ is a formula encoding all reachable states: $R(X) \land \Tr(X, X') \land \Tr^n(X',\dots,X'') \rightarrow 
    R(X) \land \Tr(X, X') \land R(X') \land \Tr^n(X',\dots,X'')$.
    Since $R(X) \land \Tr^n(X, \dots, X') \rightarrow \TrInv(X,X')$, then trivially
    $R(X) \land \Tr(X, X') \land R(X') \land \Tr^n(X',\dots,X'') \rightarrow R(X) \land \Tr(X,X') \land \TrInv(X',X'')$.
    Now, due to the fact that $R(X) \rightarrow C(X)$, the following holds:
    $R(X) \land \Tr(X, X') \land \TrInv(X',X'') \rightarrow C(X) \land \Tr(X,X') \land \TrInv(X',X'').$
    By the lemma hypothesis:
    $$C(X) \land \Tr(X,X') \land \TrInv(X',X'') \rightarrow \TrInv(X,X'').$$

    Therefore, through the chain of implications, we have:
    $$R(X) \land \Tr^{n+1}(X,X') \rightarrow \TrInv(X,X').$$

    Thus $\TrInv$ is a disjunctively well-founded transition 
    invariant over all reachable states, and $\TS$ is 
    terminating.
\end{proof}

%-----------------------------------------------------------------------
\subsection{Interpolation-based Termination Analysis}
%-----------------------------------------------------------------------

This subsection presents the full termination analysis procedure using the interpolation-based transition invariant generation.
It integrates \trInvSynthesis into the main loop of Algorithm~\ref{Alg:SNA}.
At each iteration, instead of only blocking terminating states, the algorithm also constructs a transition invariant 
candidate. 
The overall flow is shown in Figure~\ref{fig:SNATrInv}.
 
\begin{figure}[t]
    \centering
    \includegraphics[width=\linewidth]{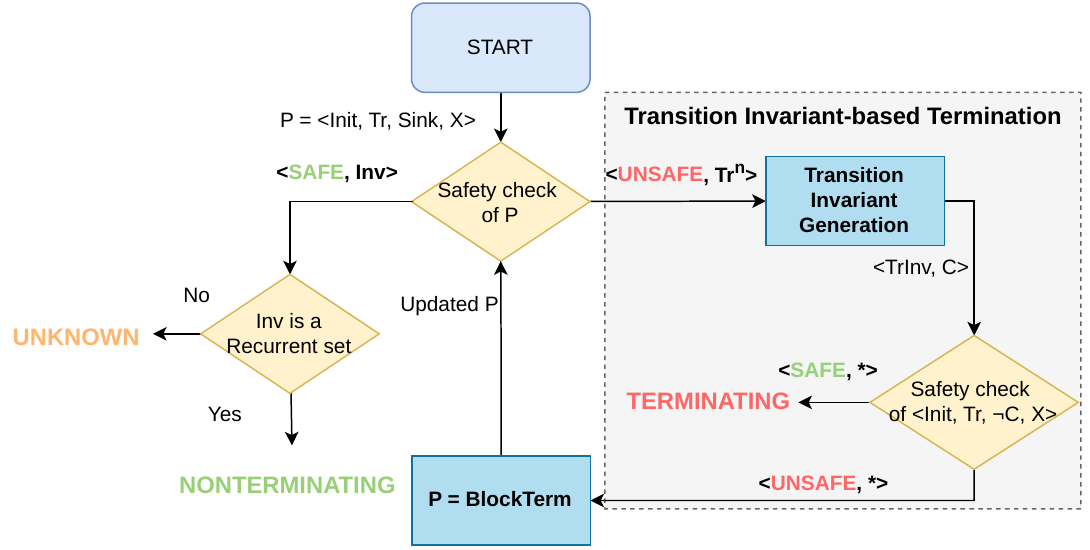}
    \caption{Execution Flow of Interpolation-based Transition Invariant Generation (New components are placed in a grey block, lines~\ref{SNA_Term:Synth},\ref{SNA_Term:CheckNc},\ref{SNA_Term:NEWTERM}).}
    \label{fig:SNATrInv}
\end{figure}

%  Notation is confusing, replace everything with phis, no need for phi'
%  On the second safety check use *
\begin{algorithm}[t!]
	\SetInd{0.5em}{0.5em}
    \DontPrintSemicolon
    \small
    \newcommand{\Continue}{\textbf{continue}}
    \caption{Interpolation-based Transition Invariant Generation for Proving and Disproving Termination.}
    \label{Alg:SNA_Term}
    \SetKwInOut{Input}{Input}
    \SetKwInOut{Output}{Output}
    \SetKwInOut{Data}{Data}
    \Input{Safety problem $P = \tuple{\Init, \Tr, \Sink, X}$}
    \Output{\TERM $\mid$ \NONTERM $\mid$ \UNKNOWN}
    \Data{$\TrInv$ - formula that encodes transition invariant}
    $\TrInv \gets \bot$ \\
    \While{$\top$} {
        $\tuple{\res, \phi = \Tr^n\ or\ \Inv} \gets \safetyCheck(P)$ \label{SNA_Term:Safety}\\
        \eIf{$\res = \UNSAFE$} {
            $\tuple{\TrInv, C} \gets \trInvSynthesis(P, \phi, \TrInv)$ \label{SNA_Term:Synth} \\
            $\tuple{\res', \_} \gets \safetyCheck(\tuple{\Init, \Tr, \lnot C, X})$ \label{SNA_Term:CheckNc} \\
            \lIf{$\res' = \UNSAFE$} {%
                $P \gets \blockTerm(P, \phi)$\label{SNA_Term:BlockTerm}
            }
            \textbf{else} \Return $\TERM$\label{SNA_Term:NEWTERM}
        }{
            \uIf{$\forall X. \phi(X) \rightarrow \exists X'. \Tr(X,X') \land \phi(X')$} {
                \Return $\NONTERM$ \label{SNA_Term:NONTERM}
            }
            \Return \UNKNOWN
        }
        
    }
\end{algorithm}

The Algorithm~\ref{Alg:SNA_Term} takes as input a safety problem $P=\tuple{\Init, \Tr, \Sink, X}$, where $\Init$ and 
$\Tr$ encode the initial states and transition relation, $\Sink$ is the formula encoding sink states, and $X$ is 
the set of variables over which the transition system is defined.
As an output, the algorithm returns \NONTERM if $\TS$ is nonterminating, \TERM if it is terminating, and \UNKNOWN if the procedure 
cannot determine the final result.
The procedure maintains a candidate transition invariant $\TrInv(X,X')$.
During the execution, $\TrInv$ is extended with new disjuncts, which are generated by the \trInvSynthesis procedure, and encode well-founded relations. 
Changes to the original SNA are at lines~\ref{SNA_Term:Synth},~\ref{SNA_Term:CheckNc}, and~\ref{SNA_Term:NEWTERM}.

Every iteration of Algorithm~\ref{Alg:SNA_Term}, similarly to the Algorithm~\ref{Alg:SNA}, 
begins with the safety check at line~\ref{SNA_Term:Safety}.
The algorithm handles two possible cases:

\paragraph{\UNSAFE} A terminating trace $\Tr^n$ exists.
The algorithm calls 
\trInvSynthesis\ to extend $\TrInv$ and compute the coverage formula $C(X)$ (line~\ref{SNA_Term:Synth}).
Then, it checks if $\lnot C(X)$ is reachable from the initial states (line~\ref{SNA_Term:CheckNc}).
If $\lnot C(X)$ is reachable, the algorithm refines the safety problem by blocking 
the states that deterministically lead to termination via \blockTerm, and the loop continues.
Otherwise, $\TrInv$ is a transition invariant over all reachable states, and the algorithm concludes termination.

\paragraph{\SAFE} No sink state is reachable, and there exists an inductive invariant $\Inv(X)$. 
If $\Inv(X)$ is a recurrent set, the system is nonterminating. 
Otherwise, the algorithm returns \UNKNOWN as the analysis is inconclusive.

\begin{example}
    Consider $\TS$ from Example~\ref{Ex:Tn} with 
    $\Init \triangleq x > 0$, 
    $\Tr \triangleq (x'=x-1 \lor x'=x-2) \land (x>0 \lor x<-10)$, 
    $X \triangleq \{x\}$.
    Preprocessing computes 
    $\Sink \triangleq x \leq 0 \land x \geq -10$.
    Then, the algorithm is called with $P=\tuple{\Init, \Tr, \Sink, X}$.

    The safety check detects a trace $\Tr^1$ reaching $\Sink$. 
    \trInvSynthesis produces a transition invariant candidate $\TrInv(X,X') \triangleq x' < x \land x > 0$ 
    and the formula $C(X) \triangleq (x \leq 0 \land x \geq -10) \lor x > 0$.
    The safety check for $\tuple{\Init, \Tr, \lnot C(X), X}$ returns \SAFE, and the algorithm concludes 
    that the transition system is terminating.

    \textbf{Limitation of the approach.} Consider a different $\TS$, 
    with $X = \{x, u\}$, $\Init(X)   \triangleq u < 10 \land x > 0 $, and
    $\Tr(X, X') \triangleq ((x' = x - 1 \land u' = u ) \lor (x' = x + 1 \land u' = u + 1 \land u < 10)) \land x > 0$.
    Intuitively, each transition either decrements $x$ (and keeps $u$ unchanged), or 
    increments both $x$ and $u$ until $u < 10$. 
    Once $u = 10$, only the decrementing transition is possible, and $x$ decreases monotonically to zero. 
    Hence $\TS$ is terminating.
    The sink states are $\Sink(X) \triangleq x \leq 0$.

    The Algorithm~\ref{Alg:SNA_Term} detects a trace $\Tr^1(X^{(0)}, X^{(1)})$ reaching $\Sink$ states.
    \trInvSynthesis\ constructs transition invariant: $\TrInv(X, X') \;\triangleq\; x' < x \land x > 0$. 
    The coverage formula is computed as: $C(X) \;\triangleq\; x \leq 0 \lor (x>0 \land u = 10)$.
    The safety check for $\tuple{\Init, \Tr, \lnot C, X}$ returns \UNSAFE, since states 
    with $u < 10$ are reachable from $\Init$. 
    \blockTerm then blocks the states that are guaranteed to reach sink:
    $\Tr(X, X') \;\gets\; \Tr(X, X') \land x' > 0$.

    The algorithm calls $\safetyCheck(P)$, which returns $\SAFE$ with inductive invariant 
    $\Inv(X) \triangleq x > 0$. 
    However, $\Inv$ is not a recurrent set, so Algorithm~\ref{Alg:SNA_Term} returns \textsc{Unknown}.

    The root cause why the algorithm cannot prove termination is that $\TrInv$ does 
    not cover states with $u < 10$. 
    It captures why the system terminates when $u = 10$, but not the global argument: $u$ 
    grows until at most $u = 10$, after which $x$ decreases monotonically. 
    A complete termination proof requires a \emph{second} disjunct covering the $u < 10$ region.
    In Section~\ref{Sec:Extension}, we present an extension that focuses on the state space
    for which the transition invariant candidate is not valid, aiming to extend the transition invariant.
\end{example}

%-----------------------------------------------------------------------
\subsection{Correctness}
%-----------------------------------------------------------------------

This subsection proves the soundness of the introduced approach. We first prove that when the algorithm returns \TERM, the transition system is terminating,
and then we prove the correctness for the nontermination.
%  Add well-foundness definition

\begin{lemma}\label{Lem:TERM_N}
    If Algorithm~\ref{Alg:SNA_Term} returns $\TERM$, then the transition system $\tuple{\Init, \Tr, X}$ is terminating.
\end{lemma}
\begin{proof}
    $\TERM$ is returned only at line~\ref{SNA_Term:NEWTERM}.
    It means that there exists a transition invariant $\TrInv$ which is valid over states $C(X)$, and $\lnot C(X)$ is 
    unreachable from the initial states.
    The safety problem is updated iteratively, blocking the states that deterministically lead to termination,
    so in the end the $P = \tuple{\Init', \Tr', \Sink, X}$, where $\Init'(X) = \Init(X) \land \bigwedge_{i=0}^{m} \lnot S_i(X)$
    and $\Tr' = \Tr(X,X') \land \bigwedge_{j=0}^{l} \lnot S_j(X')$ .
    According to Lemma~\ref{Lem:Term}, transition system $\tuple{\Init', \Tr', X}$ is terminating, since $\TrInv$ is a disjunctively well-founded
    transition invariant over all reachable states.

    Since the states satisfying $\bigvee_{i=0}^{m} S_i(X)$ and $\bigvee_{j=0}^{l} S_j(X')$ deterministically lead to the termination, 
    every trace of $\tuple{\Init, \Tr, X}$ either follows $\Tr'$ or passes through a blocked state, which leads to termination.
    Hence $\tuple{\Init, \Tr, X}$ is terminating.
\end{proof}

\begin{lemma}
    If Algorithm~\ref{Alg:SNA_Term} returns $\NONTERM$, then the transition system $\tuple{\Init, \Tr, X}$ is nonterminating.
\end{lemma}
\begin{proof}
    The proof of this lemma follows directly from Lemma~\ref{Lem:NONTERM}, 
    as the nontermination analysis in Algorithm~\ref{Alg:SNA_Term} 
    follows exactly the same process as Algorithm~\ref{Alg:SNA}.
\end{proof}

\section{Extended Interpolation-based Transition Invariant Generation for Proving and Disproving Termination}\label{Sec:Extension}

Algorithm~\ref{Alg:SNA_Term} is already capable of proving termination 
efficiently through the construction of transition invariants.
This section describes an extension of the algorithm that aims to improve interaction between termination and nontermination analyses.
Additionally, we present further extensions to the original SNA algorithm that improve overall performance. 
The high-level flow of the extended algorithm can be seen in Figure~\ref{fig:SNA++}.

\subsection{Analysis of States for which Transition Invariant is not Valid}~\label{Sec:Refinement}

A transition invariant candidate can be used to guide the termination analysis.
Particularly, Algorithm~\ref{Alg:SNA_Term} constructs the formula $C(X)$ encoding states for which the transition invariant is valid.
If $\lnot C(X)$ is reachable, there exist states for which the transition invariant is not valid, and which potentially can lead to nontermination.
The intuition is to compute the reachable states $S_m$ within $\lnot C(X)$ (which corresponds to line~\ref{SNA_Refined:Sm} of Algorithm~\ref{Alg:SNA_Refined}).
Then, it is possible to run (non)termination check for these particular states within $\TS$ (line~\ref{SNA_Refined:PushNewProblem}).

In the case of termination, the algorithm produces a witness transition invariant that extends the 
transition invariant of the original system, making it valid for more states.
If nontermination is proved for $S_m$, then the original $\TS$ is also nonterminating, since $S_m$ is 
reachable from the initial states.

\begin{figure}[t]
    \centering
    \includegraphics[width=\linewidth]{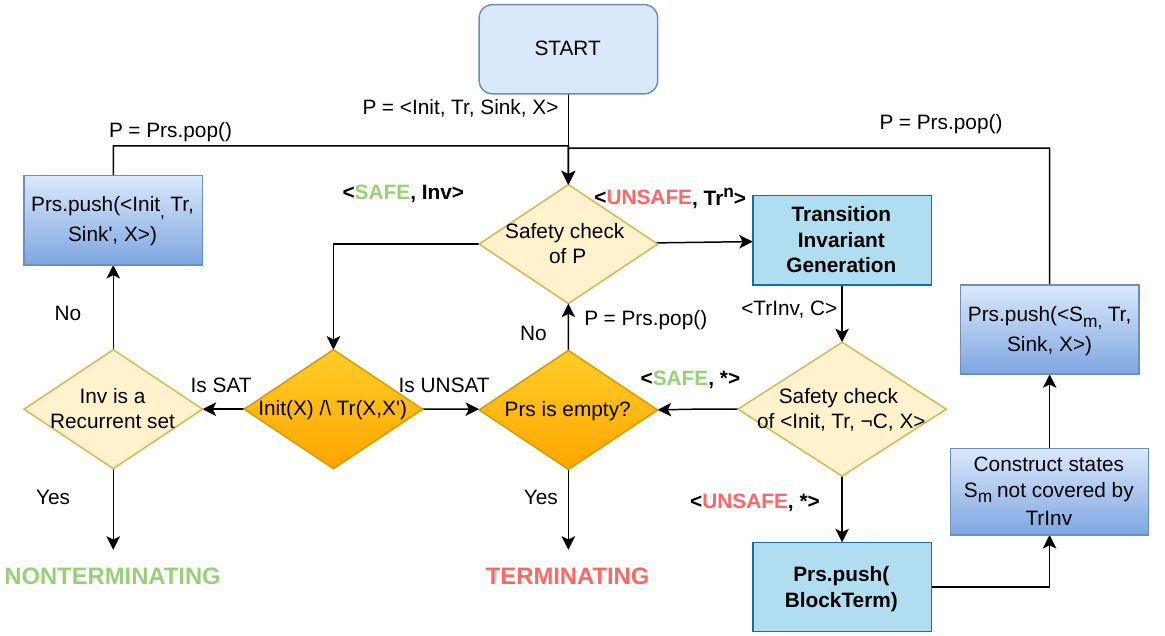}
    \caption{Extended Interpolation-based Termination Analysis (new components are highlighted in dark orange and dark blue colors).}
    \label{fig:SNA++}
\end{figure}

\subsection{Overall Algorithm}

% Mention that A and B can be used for Algo 1
% Make people understand that C is much more important!!! Example to motivate each section!!!
% The optimized version of the algorithm is similar to the Algorithm~\ref{Alg:SNA_Term}.
% Main differences are the refinement of the safety problem, and the check for trivial termination.
% These changes allow to significantly improve the overall power and performance of the algorithm.
% In the description, we would concentrate particularly on the changes in the flow, as the rest of the algorithm is quite similar.
% Description in the line10 - maybe explicit else - continue.

\begin{algorithm}[t!]
	\SetInd{0.5em}{0.5em}
    \DontPrintSemicolon
    \small
    \newcommand{\Continue}{\textbf{continue}}
    \caption{Extended Interpolation-based Transition Invariant Generation for Proving and Disproving Termination.}
    \label{Alg:SNA_Refined}
    \SetKwInOut{Input}{Input}
    \SetKwInOut{Output}{Output}
    \SetKwInOut{Data}{Data}
    \Input{Safety problem $P$}
    \Output{\TERM $\mid$ \NONTERM}
    \Data{$\TrInv$ - candidate transition invariant, $\Prs$ - stack of safety problems}
    $\TrInv \gets \bot$ \\
    $\Prs \gets [P]$ \label{SNA_Refined:Prs} \\
    \While{$\top$} {
        \label{SNA_Refined:Extension} $P' \equiv \tuple{\Init, \Tr, \Sink, X} \gets \Prs.pop()$ \\
        \label{SNA_Refined:Safety} $\tuple{\res, \phi} \gets \safetyCheck(P')$ \\
        \eIf{$\res = \UNSAFE$} {
            $\tuple{\TrInv, C} \gets \trInvSynthesis(P', \phi, \TrInv)$ \\
            \label{SNA_Refined:CheckNc} $\tuple{\res', \phi'} \gets \safetyCheck(\tuple{\Init, \Tr, \lnot C, X})$ \\
            \eIf{$\tuple{\res', \phi'} = \tuple{\UNSAFE, \Tr^{m}}$} 
            {
                $\Prs.push(\blockTerm(P', \phi))$ \label{SNA_Refined:BlockTerm}\\
                $S_m(X) \gets QE(\exists X^{(0)}, \dots, X^{(m-1)}. \Init(X^{(0)}) \land \phi'(X^{(0)}, \dots, X^{(m)}) \land \lnot C(X^{(m)}))  [X^{(m)} \mapsto X]$ \label{SNA_Refined:Sm}\\
                $\Prs.push(\tuple{S_m, \Tr, C, X})$\label{SNA_Refined:PushNewProblem}
            }{
                \lIf {$\Prs = \emptyset$} {\Return $\TERM$\label{SNA_Refined:NEWTERM}}
                \textbf{else} \Continue
            } 
        }{
        %    \label{SNA_TERM:SAFE}
            \uIf{$\forall X. \phi(X) \rightarrow \exists X'. \Tr(X,X') \land \phi(X')$} {
                \Return $\NONTERM$ \label{SNA_Refined:RECURRENT}
            } 
            \uIf(\label{SNA_Refined:TERM}) {\UNSAT ? $\Init(X) \land \Tr(X,X')$} {
                    \lIf{$\Prs = \emptyset$} {\Return $\TERM$ }
                    \textbf{else} \Continue
            } {
                $\Sink'(X) \gets \lnot QE(\exists X'. \Tr(X,X'))$ \label{SNA_Refined:Refinement} \\ 
                $\Prs.push(\tuple{\Init, \Tr, \Sink', X})$ \label{SNA_Refined:NewPush}
            }
        }
        
    }
\end{algorithm}

Algorithm~\ref{Alg:SNA_Refined} extends Algorithm~\ref{Alg:SNA_Term} to enable more efficient (non)termination analysis.
First, in addition to maintaining the transition invariant,
the algorithm maintains a stack of safety problems to be analyzed: $\Prs$ (line~\ref{SNA_Refined:Prs}).
At each iteration, the algorithm pops a safety problem from the stack and analyzes it.
The stack allows the algorithm to focus on (non)termination analysis of particular reachable states within 
$\TS$.
The original safety problem always remains in the $\Prs$ stack.

The postprocessing phase of transition invariant construction is the main change to Algorithm~\ref{Alg:SNA_Term}.
When $\lnot C$ is reachable, the extended algorithm performs three actions: (1) it pushes the $\blockTerm$-refined problem 
(line~\ref{SNA_Refined:BlockTerm}); (2) it computes a formula encoding reachable states,
for which transition invariant candidate is not valid (line~\ref{SNA_Refined:Sm}); and (3) it pushes a new 
subproblem to analyze for termination (line~\ref{SNA_Refined:PushNewProblem}), with $S_m$ as the initial states. 
This allows the algorithm to concentrate on proving termination or nontermination of particular reachable states.

If $\lnot C$ is not reachable, $\TrInv$ proves termination of the currently analyzed $\TS$ (the last $\TS$ 
popped from the stack), and the number of remaining subproblems decreases.
If the stack is empty, Algorithm~\ref{Alg:SNA_Refined} returns \TERM.
The algorithm incorporates two further extensions.

\textit{Simple termination check (line~\ref{SNA_Refined:TERM}).} 
When the safety check returns $\SAFE$ and $\phi(X)$ is not a recurrent set, 
the algorithm additionally checks whether it is possible to take a transition 
from the initial states by querying:
$$\Init(X) \land \Tr(X,X')$$
If this formula is unsatisfiable, then no transition can be taken from the initial states 
(or all transitions are blocked, meaning all reachable states deterministically lead to termination), 
and the system is trivially terminating.

\textit{Refinement of $\Sink$ states (lines~\ref{SNA_Refined:Refinement},~\ref{SNA_Refined:NewPush}).} 
When $\phi(X)$ is not a recurrent set, this is often because \blockTerm\ 
has eliminated all transitions leading to $\Sink$.
This means that the produced $\phi$ contains sink states due to the updated $\Tr$. 
New formula encoding sink states is computed as:

$$\Sink'(X) = \lnot QE(\exists X'. \Tr'(X,X'))$$

Formula $\Sink'$ replaces $\Sink$ in the safety problem, enabling the analysis to 
proceed instead of returning \UNKNOWN. 

\begin{lemma}\label{Lem:FinalTermination}
    If Algorithm~\ref{Alg:SNA_Refined} returns $\TERM$, then the original transition system $\tuple{\Init, \Tr, X}$ is terminating.
\end{lemma}
\begin{proof}
    \TERM is returned at line~\ref{SNA_Refined:NEWTERM} or line~\ref{SNA_Refined:TERM}.
    As in Algorithm~\ref{Alg:SNA_Term}, to return $\TERM$, the termination conditions need to 
    be satisfied for the transition system $\tuple{\Init', \Tr', X}$, where $\Init'(X) = \Init(X) \land \bigwedge_i \lnot S_i(X)$, 
    and $\Tr'(X,X') = \Tr(X,X') \land \bigwedge_j \lnot S_j(X')$.
    $\bigvee_{i=0}^m S_i(X)$ and $\bigvee_{j=0}^l S_j(X')$ are the states that deterministically lead to termination.

    \textbf{Case 1.} If $\TERM$ is returned at line~\ref{SNA_Refined:NEWTERM}, then the termination proof 
    follows directly from Lemma~\ref{Lem:TERM_N}. This is the case, since the original safety problem is
    preserved in the stack, so if $\Prs = \emptyset$, the original safety problem is terminating.

    \textbf{Case 2.} If $\TERM$ is returned at line~\ref{SNA_Refined:TERM}, then formula 
    $\Init'(X) \land \Tr'(X,X')$ is unsatisfiable, so no transition is reachable from $\Init'$.
    Thus, all feasible transitions in the original 
    transition system lead to deterministically terminating states, and the transition system is terminating.
\end{proof}

\begin{lemma}\label{Lem:FinalNontermination}
     If Algorithm~\ref{Alg:SNA_Refined} returns $\NONTERM$, then the original transition system $\tuple{\Init, \Tr, X}$ is nonterminating.
\end{lemma}
\begin{proof}
    \NONTERM is returned at line~\ref{SNA_Refined:RECURRENT}, where 
    $\phi(X)$ is a recurrent set for the currently considered $\TS$  $\tuple{\Init', \Tr', X}$. 
    By Lemma~\ref{Lem:NONTERM}, there exists a recurrent set in the transition system $\tuple{\Init', \Tr', X}$, and since 
    $\Tr'(X,X') \rightarrow \Tr(X,X')$ and $\Init'$ is reachable from  $\Init$ (or it is $\Init$), this recurrent set witnesses nontermination of $\TS$.
\end{proof}

\section{Implementation and Evaluation}
\label{Sec:Eval}

\newcommand{\golem}{\textsc{Golem}\xspace}
\newcommand{\opensmt}{\textsc{OpenSMT}\xspace}
\newcommand{\sna}{\textsc{SNA}\xspace}
\newcommand{\snaTerm}{\textsc{ItpTig}\xspace}
\newcommand{\snaOpt}{\textsc{ItpTig+}\xspace}
\newcommand{\LoAT}{\textsc{LoAT}\xspace}
\newcommand{\KoAT}{\textsc{KoAT}\xspace}
\newcommand{\Tt}{\textsc{T2}\xspace}

This section presents the implementation details and experimental 
evaluation of our approach.

\subsection{Implementation}

We implemented the prototypes of Algorithms~\ref{Alg:SNA}, \ref{Alg:SNA_Term}, and \ref{Alg:SNA_Refined} inside of \golem; 
we refer to these implementations as \sna, \snaTerm, and \snaOpt.
\golem was chosen because it supports multiple different safety verification procedures,
such as Property-Directed Reachability~\cite{DBLP:journals/fac/BradleyM08} (PDR), Transition Power Abstraction~\cite{DBLP:conf/tacas/BlichaFHS22} (TPA), and Interpolation-based Model Checking~\cite{DBLP:conf/cav/McMillan03}.
Additionally, we used quantifier elimination, already implemented in \golem.
% \snaOpt with different safety verification engines was able to solve different 
% termination problems uniquely.
% Quantifier Elimination is produced by executing Model Based Projections, untill the whole formula is covere
%
For SMT queries and interpolation, our implementation uses the \opensmt  solver~\cite{DBLP:conf/sat/HyvarinenMAS16,DBLP:conf/tacas/BruttomessoPST10}.
\opensmt is integrated within \golem and provides native support for interpolation.

\subsection{Evaluation}

We compared \snaOpt against \sna and \snaTerm to evaluate the effect of combining termination and nontermination analyses.
Additionally, we evaluated \snaOpt against several state-of-the-art tools: 
\LoAT~\cite{10.1007/978-3-031-90660-2_13} (version eacb74b), \KoAT~\cite{10.1007/978-3-031-90660-2_13} 
(version 41cb681), \Tt~\cite{DBLP:conf/tacas/BrockschmidtCIK16} (version 10f1373). 
The evaluation focused on termination analysis of Integer Transition Systems. 
Experiments were conducted on Ubuntu 20.04 machine with an AMD EPYC 7452 32-core processor and 8x32 GiB RAM. 
Our evaluation addressed the following research questions:
\begin{itemize}
\item \textbf{RQ1:} How does extending \sna with the termination procedure affect verification performance?
\item \textbf{RQ2:} How does \snaOpt perform compared to state-of-the-art tools?
\end{itemize}

The benchmark suite of Integer Transition Systems was taken from the Termination Competition~\cite{DBLP:conf/cade/GieslMRTW15}. 
Overall, there are 1222 benchmarks (538 nonterminating, 665 terminating, and 19 unknown). 
We excluded 42 trivially terminating benchmarks from the evaluation (termination of which can be detected by trivial syntactic checks), which
are solved by all tools.
The evaluation was executed with the 120 seconds timeout (the increase of the timeout does not 
significantly influence the evaluation).

\begin{table}[t!]
    \centering
    \caption{Comparison of \snaOpt with \sna and \snaTerm.}
    \label{tab:internal}
    \begin{tabular}{|c|c|c|c|c|}
        \hline
        \textbf{Algorithm} & \textbf{Total} & \textbf{Terminating} & \textbf{Nonterminating} & \textbf{Uniquely} \\
        \hline
        \sna & 343 & 0 & 343 & 7 \\
        \hline
        \snaTerm & 521 & 186 & 335 & 7 \\
        \hline
        \snaOpt & 761 & 320 & 441 & 240 \\
        \hline
    \end{tabular}
\end{table}
% \begin{figure}[t]
%     \centering
%     \includegraphics[width=1\linewidth]{Figures/TermCompCMP.pdf}
%     \caption{Comparison of \snaOpt with \sna and \snaTerm. The left plot is solved terminating instances, the right is nonterminating.}
%     \label{fig:golemComp}
% \end{figure}

Table~\ref{tab:internal} gives the experimental results for \snaOpt, \sna, and \snaTerm. 
It confirms that \snaOpt consistently outperforms both \sna and \snaTerm for terminating 
and nonterminating instances, solving 240 instances solved by neither \sna nor \snaTerm. 
We attribute this to \snaOpt's generation of focused sub-problems, 
which allows the analysis to concentrate on specific reachable states.
This benefits recurrent set detection by focusing the analysis on specific states not covered by the transition invariant.
It also benefits termination proving by extending the transition invariant with additional disjuncts 
that cover different regions of the reachable state space.
In answer to \textbf{RQ1}, the combination of termination and nontermination reasoning in 
\snaOpt significantly outperforms either individual approach in isolation.

Interestingly, \snaOpt behaves differently depending on the underlying safety verification engine:
with TPA it solved 761 instances and
with PDR -- 755.
That said, it solved 24 (resp. 18) benchmarks uniquely with TPA (resp. PDR).

\begin{figure*}[t!]
    \centering
    \includegraphics[width=1\linewidth]{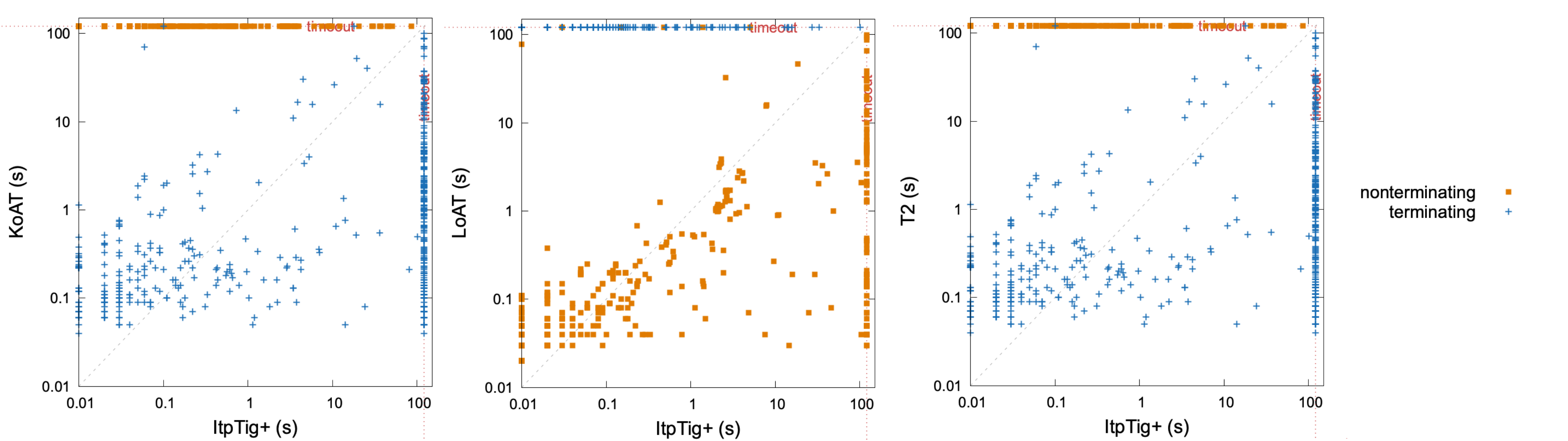}
    \caption{Comparison of \snaOpt with \KoAT, \LoAT and \Tt (runtime in seconds).}
    \label{fig:stateOfArtComp}
\end{figure*}

\begin{table}[t!]
    \centering
        \caption{Comparison of \snaOpt with state-of-the-art tools.}
    \label{tab:external}
    \begin{tabular}{|c|c|c|c|c|}
        \hline
        \textbf{Algorithm} & \textbf{Total} & \textbf{Terminating} & \textbf{Nonterminating} & \textbf{Uniquely} \\
        \hline
        \LoAT & 525 & 0 & 525 & 48 \\
        \hline
        \KoAT & 562 & 562 & 0 & 26 \\
        \hline
        \Tt & 1002 & 572 & 430 & 40 \\
        \hline
        \snaOpt & 761 & 320 & 441 & 8 \\
        \hline
    \end{tabular}
\end{table}

Figure~\ref{fig:stateOfArtComp} shows a scatter plot comparing \snaOpt against 
state-of-the-art termination analyzers, where each axis represents the solving time in seconds for a given tool.
Orange squares denote nonterminating instances; blue crosses denote terminating instances.
Points above the diagonal indicate instances where \snaOpt is faster than the competitor.
\snaOpt solves 761 benchmarks, with \Tt solving 1002, \KoAT 562, and \LoAT 525 instances, the full comparison can be seen in Table~\ref{tab:external}.
As shown in Figure~\ref{fig:stateOfArtComp}, \snaOpt, despite being a prototype implementation, is competitive with state-of-the-art termination analyzers.
Among all compared tools, \snaOpt uniquely solves eight instances.
Notably, two of these — \texttt{juLinkedListCreateAddAll.jar-obl-11} and \texttt{juLinkedListCreateAddAllAt.jar-obl-17} have never been solved in the Termination Competition.
\snaOpt found recurrent sets and proved nontermination for them.
In answer to \textbf{RQ2}, \snaOpt is competitive with state-of-the-art tools and solves 
instances not solved by any tool in the Termination Competition.

\section{Related Work}
\label{Sec:Related}
Recent years have seen significant advances in automated techniques for both termination and nontermination analysis. 
%Below we discuss the techniques related to our approach.

\textbf{Nontermination.} 
Nontermination analysis is typically centered on the construction 
of recurrent sets — sets of states from which execution cannot 
escape~\cite{DBLP:conf/cade/FrohnG22,DBLP:conf/tacas/BrockschmidtCIK16,
DBLP:conf/tacas/ChenCFNO14,DBLP:conf/foveoos/BrockschmidtSOG11,
DBLP:conf/cav/LarrazNORR14,FedyukovichZG18}.
These techniques range from template-based synthesis to 
safety-check-based recurrent set construction, the latter of which 
our approach directly extends~\cite{DBLP:conf/tacas/ChenCFNO14}.
An alternative approach detects nontermination by identifying 
lasso-shaped executions — finite paths leading to a cycle — via 
bounded model checking~\cite{DBLP:journals/entcs/BiereAS02}.
Leike and Heizmann~\cite{DBLP:conf/tacas/LeikeH18} use 
geometric nontermination witnesses: finite representations of 
infinite executions via sums of geometric series.
Unlike the approaches above, which treat nontermination in isolation, 
our technique is integrated with termination analysis within a unified 
framework, so that intermediate results guide one another.

%  within a unified framework
% Nontermination analysis is typically centered around the construction of recurrent sets — sets of states from which execution cannot 
% escape~\cite{DBLP:conf/sas/Ben-AmramDG19,DBLP:conf/cade/FrohnG22,DBLP:conf/tacas/BrockschmidtCIK16,DBLP:conf/tacas/ChenCFNO14,DBLP:conf/foveoos/BrockschmidtSOG11,DBLP:conf/cav/LarrazNORR14}.
% These techniques can vary from the template-based synthesis to the safety-check based recurrent set construction.
% Another traditional approach is safety-based nontermination analysis~\cite{DBLP:journals/entcs/BiereAS02},
% that detects nontermination if the same state is reachable within TS more then once.
% Another interesting approach is the construction of geometric nontermination argument~\cite{DBLP:conf/tacas/LeikeH18} -
% a finite representation of an infinite execution that via a sum of geometric series.

\textbf{Termination.} 
The most widely studied approach to termination analysis is the 
synthesis of ranking functions~\cite{DBLP:conf/vmcai/PodelskiR04,
DBLP:conf/cav/KuraUH20,DBLP:journals/jacm/Ben-AmramG14,
DBLP:journals/jar/GieslABEFFHOPSS17,DBLP:conf/esop/SaritaSGSV26,
DBLP:conf/tacas/UrbanGK16,FedyukovichZG18,RileyF25}.
State-of-the-art techniques focus on the construction of 
multiphase or lexicographic ranking 
functions~\cite{DBLP:journals/jacm/Ben-AmramG14,FedyukovichZG18}, which 
significantly broaden the class of programs for which termination 
can be proved. However, the complexity of a transition system can 
make ranking function synthesis infeasible.
With fixpoint computation, ~\cite{DBLP:journals/pacmpl/UnnoTGK23} reduces termination and nontermination to the validity of a first-order fixpoint formula.
Disjunctively well-founded transition 
invariants~\cite{DBLP:conf/lics/PodelskiR04,DBLP:conf/sas/CookPR05,
DBLP:conf/cav/KroeningSTW10,DBLP:conf/tacas/TsitovichSWK11} enable over-approximating the behavior of the transition 
system. % reducing the termination proof to a disjunctive 
%well-foundedness check on the invariant. 
The central challenge there is the generation of the 
transition invariant itself: existing techniques rely either on 
template-based generation~\cite{DBLP:conf/tacas/TsitovichSWK11} or on the direct construction of a 
ranking relation from a particular trace~\cite{DBLP:conf/sas/CookPR05}. 
In contrast, our approach 
uses Craig interpolation to guide transition invariant synthesis 
with respect to the actual termination conditions, producing 
property-directed invariants.
%  from terminating traces. 

% There also exists another set of techniques that are based on the construction of well-founded transition
% invariants~\cite{DBLP:conf/lics/PodelskiR04,DBLP:conf/sas/CookPR05,DBLP:conf/cav/KroeningSTW10,DBLP:conf/tacas/TsitovichSWK11}.
% These approaches overapproximate the behaviour of the transition system, simplifying the termination proof.
% The core problem of these approaches is the generation of transition invariant.
% Listed techniques rely either on the template-based generation or on the direct generation of ranking relation based on a particular trace
% in the transition system.
% In contrast, we propose to do the guided generation of transition invariants, based on the terminating traces and the actual 
% termination conditions.
% Additionally, our technique is combined with a nontermination approach, increasing the efficiency of both techniques.

%

\section{Conclusion}
\label{Sec:Conclusion}

% This paper introduces Interpolation-based Transition Invariant Generation for Proving and Disproving Termination.
% A novel approach for the termination analysis that extends Safety-based Nontermination Analysis technique.
% Our key insight is to use interpolation for the construction of termination witnesses,
% that allows the transition invariant to be guided by termination condition.
% Additionally, our technique is strongly centered on the interaction between termination and nontermination
% approaches, using transition invariants to limit the explored state space and using safety-based 
% nontermination to generate candidate transition invariants.
% Application of safety verifiers provides a strong background for the analysis, allowing to 
% quickly check termination and nontermination.
% We proved the soundness of the algorihtm and conducted an extensive experimental evaluation, 
% which demonstrates both the eficiency of the combination of termination and nontermination techniques,
% and the competitive  performance compared to other state-of-the-art termination analysis tools.

We have presented an interpolation-based approach to transition 
invariant generation for proving and disproving termination of 
infinite-state systems. 
The approach extends safety-based nontermination analysis with a termination proving capability, 
unifying the two within a single framework. 
Our key insight is to exploit terminating traces produced by the SNA
as a source of structural information, using Craig interpolation 
to construct well-founded transition invariants guided by the 
termination conditions. 
The interaction between termination and nontermination reasoning is central to our technique: 
transition invariants focus the nontermination analysis on not covered states, 
while SNA provides terminating traces that drive invariant generation.
Grounding the approach in safety verification provides a robust foundation, enabling 
efficient checking of both termination and nontermination via 
existing safety solvers. 
We proved the soundness of the algorithm and conducted an extensive
experimental evaluation. 
The results demonstrate the benefit of combining termination and nontermination analysis, 
and show that our approach achieves competitive performance against state-of-the-art tools, 
including the unique solution of instances never solved by any 
tool in the Termination Competition.

\textbf{Acknowledgements.} This work was supported in part by SNSF grant No. 200021-236601.
We thank Marek Jankola and Dirk Beyer for fruitful discussions 
on termination analysis. 

\bibliographystyle{IEEEtran}
\bibliography{biblio}

\end{document}